\documentclass[12pt]{article}

\usepackage{epsfig}

\usepackage{amssymb}
\usepackage{amsfonts}

\usepackage{color}
 
\def\be{\begin{equation}}
\def\ee{\end{equation}}
\def\ba{\begin{array}{c}}
\def\ea{\end{array}}

\newcommand{\bea}{\begin{eqnarray}}
\newcommand{\eea}{\end{eqnarray}}

\newcommand{\kt}{\rangle}
\newcommand{\br}{\langle}

\newtheorem{thm}{Theorem}

\newtheorem{lemma}[thm]{Lemma}

\newenvironment{proof}{\noindent
 {\bf Proof.}}{\hfill$\square$\vspace{3mm}\endtrivlist}

\begin{document}

\begin{center}

.

{\Large \bf

Hybrid form of quantum theory
with non-Hermitian Hamiltonians

 }

\vspace{10mm}

\vspace{0.2cm}

\vspace{2mm}

\textbf{Miloslav Znojil}

\vspace{0.2cm}

\vspace{0.2cm}

Department of Physics, Faculty of Science, University of Hradec
Kr\'{a}lov\'{e},

Rokitansk\'{e}ho 62, 50003 Hradec Kr\'{a}lov\'{e},
 Czech Republic

\vspace{0.2cm}

 and

\vspace{0.2cm}

The Czech Academy of Sciences, Nuclear Physics Institute,

 Hlavn\'{\i} 130,
250 68 \v{R}e\v{z}, Czech Republic

\vspace{0.2cm}

{e-mail: znojil@ujf.cas.cz}

\end{center}

\newpage

\section*{Abstract}

In Schr\"{o}dinger picture the unitarity of evolution is usually
guaranteed by the Hermiticity of the Hamiltonian operator
$\mathfrak{h}=\mathfrak{h}^\dagger$ in a conventional Hilbert space
${\cal H}_{textbook}$. After a Dyson-inspired
operator-transformation (OT) non-unitary preconditioning $\Omega:
\mathfrak{h} \to H$ the simplified Hamiltonian $H$ is, in its
manifestly unphysical Hilbert space ${\cal H}_{auxiliary}$,
non-Hermitian. Besides its natural OT-based physical interpretation
it can also be ``Hermitized'' (i.e., made compatible with the
unitarity) via a metric-amendment (MA) change of the Hilbert space,
${\cal H}_{auxiliary}\to {\cal H}_{physical}$. In our present letter
we propose another, third, hybrid form (HF) of the Hermitization of
$H$ in which the change involves, simultaneously, both the
Hamiltonian and the metric. Formally this means that the original
Dyson map is assumed factorizable, $\Omega=\Omega_M\Omega_H$. A key
practical advantage of the new HF approach lies in the
model-dependent adaptability of such a factorization. The
flexibility and possible optimality of the balance between the
MA-related (i.e., metric-amending) factor $\Omega_M$ and the
OT-related (i.e., Hamiltonian-changing) factor $\Omega_H$ are
explicitly illustrated via an elementary two-state quantum model.


\section*{Keywords}
.

non-Hermitian quantum mechanics of unitary systems;

hiddenly Hermitian quantum Hamiltonians;

factorized Dyson map;

Hermitization using a combined amendment of the inner product and Hamiltonian;

\newpage

\section{Introduction}



In the Dyson's paper \cite{Dyson} the author had to deal with a
specific realistic quantum Hamiltonian
 \be
 \mathfrak{h}=\mathfrak{h}^\dagger\,
 \label{hepo}
 \ee
which was defined as acting
and self-adjoint in a conventional Hilbert space ${\cal
H}_{textbook}$,
During calculations such a representation of the system appeared
``user-unfriendly'' and, for this reason, Schr\"{o}dinger equation
had to be simplified. For the purpose, the author made use of an
isospectral preconditioning of the Hamiltonian,
 \be
 \mathfrak{h}\ \to \ H = \Omega^{-1}\,\mathfrak{h}\,\Omega
 \,.
 \label{loss}
 \ee
A remarkable as well as not quite expected success has been achieved
via an insightful choice of an unusual, {\it manifestly
non-unitary\,} preconditioning operator
 \be
 \Omega_{} \neq
 \left ( \Omega_{}^{-1} \right)^\dagger\,.
 \label{Dmap}
 \ee
The original Schr\"{o}dinger equation containing the prohibitively
complicated initial Hamiltonian
$\mathfrak{h}=\mathfrak{h}_{(Dyson)}$ appeared converted into a more
easily tractable evolution and eigenvalue problem.

In {\it loc. cit.}, in spite of the related loss of
the Hermiticity of $H_{}\neq H_{}^\dagger$, the
invertible operator-transformation (OT) change (\ref{loss})
of the Hamiltonian
based on a judicious choice or construction of a suitable Dyson map (\ref{Dmap})
paid off. The first step towards the
innovative use of non-Hermitian Hamiltonians $H \neq H^\dagger$
with real spectra in the quantum mechanics of unitary systems
has been made (see Refs.~\cite{book,MZbook} and also
a few more detailed explanatory comments in section \ref{sectiona} below).

Unfortunately,
a wider acceptance of the OT idea appeared slowed down by the
numerous emerging mathematical obstacles.
In a way emphasized, almost simultaneously,
by mathematicians \cite{Dieudonne},
the manifest non-Hermiticity of $H$ may
cause a serious conceptual concern because it often
implies, among other,
the emergence of instabilities
of the spectrum caused even by
very small
perturbations of the Hamiltonian \cite{Trefethen}.

One of the remedies has been proposed, in 1992, by Scholtz et al \cite{Geyer}.
In essence, these authors
emphasized another, alternative
(viz., ``metric-amendment'', MA)
aspect of the Dyson's concept of the Hermitization
(see also a concise outline of the MA theory
in section \ref{sectionb} below).
From the new perspective the above-mentioned
mathematical obstacles have been shown ``removable''
via additional assumptions like, typically, via the
purely technically motivated requirement that the admissible
non-Hermitian Hamiltonian operators $H$ have to be required bounded.

Naturally, the latter requirement (which proved not too essential, say,
for the model-building in
nuclear physics \cite{Jenssen})
appeared hardly acceptable in a broader quantum-theoretical context.
Fortunately, a truly sophisticated way out of the dead end
has been found by Bender with Boettcher \cite{BB}.
In a way not unknown in the older mathematical literature \cite{BG}
these authors shifted emphasis, in effect, from the mathematical,
OT-mediated correspondence (\ref{loss}) between
$H$ and its isospectral partner
$\mathfrak{h}$
to the underlying, MA-mediated physics.
In particular, there authors
turned attention
to the concept of parity-times-time-reversal symmetry of $H$
(${\cal PT}-$symmetry of $H$).
As a consequence their study
inspired the current
enormous
popularity of the widespread
use of the manifestly non-Hermitian Hamiltonians $H$ in
the various areas of physics
(see, e.g., the
early paper \cite{alien}, the
standard reviews \cite{Carl,ali} or the more
recent monographs \cite{book,Christodoulides}).

In our present letter we intend to complement these developments
by an innovative proposal of a certain third, ``hybrid-form'' (HF) combination
of the two standard and apparently distinct OT and MA theory-building strategies.
In a way explained in section \ref{sectionc} and illustrated by
a schematic application in section \ref{sectiond}
we will show that
for a unitarily evolving quantum system characterized by
a preselected phenomenological
Hamiltonian $H$ (with real spectrum)
an optimal representation of its unitary evolution
might be based on a mere partial
(i.e., not too complicated)
modification of the Hamiltonian operator
accompanied by a mere partial
(i.e., not too complicated)
amendment of the related correct physical Hilbert-space metric.

\section{The Dyson-inspired Hermitization $H \ \to \  \mathfrak{h}$\label{sectiona}}

In the early nineties Scholtz et al \cite{Geyer} recalled
the
success of the isospectral
simplification by mapping (\ref{loss}) in nuclear physics \cite{Jenssen}.
In its light they proposed an
extension of the theory and a
replacement of the de-Hermitizing mapping $\Omega: \mathfrak{h} \to H$ by
its Hermitizing inversion,
 \be
 H \ \to \  \mathfrak{h} = \Omega\,H\,\Omega^{-1}=
 \mathfrak{h}^\dagger\,.
 \label{gain}
 \ee
This means that one could consider certain less usual quantum
models which are
controlled by a manifestly non-Hermitian candidate $H$ for the Hamiltonian.
The transition (\ref{gain}) to its
isospectral self-adjoint partner $\mathfrak{h}$ can be then interpreted
as an OT-mediated Hermitization and a guide to the standard probabilistic interpretation
of the system in question.
The unitarity of the evolution of the underlying
quantum system as required by the Stone theorem \cite{Stone}
in Schr\"{o}dinger picture would then be restored.

In this spirit an innovative model-building process can
start from any sufficiently user-friendly
non-Hermitian candidate $H$ for the Hamiltonian.
Still,
a necessary guarantee of
the unitarity of the evolution must be provided,
but this can be achieved by the mere
reference to the isospectrality (\ref{gain}) between $H$
and its
self-adjoint partner
$\mathfrak{h}$ \cite{Geyer}.
What is obtained is
an apparently non-Hermitian
or, in the current mathematical terminology, quasi-Hermitian \cite{Dieudonne}
formulation of quantum mechanics (QHQM).
According to its reviews in the newer literature
the QHQM theory can be also assigned alternative names of
${\cal PT}-$symmetric quantum mechanics
\cite{Carl},
three-Hilbert-space quantum mechanics
\cite{SIGMA} or
pseudo-Hermitian quantum mechanics
\cite{ali}.

In the latter, quantum-mechanics-representation context
let us emphasize
that in the three reviews \cite{Geyer,Carl} and \cite{ali}
(as well as in our present forthcoming text, for the sake of brevity)
the Dyson map itself
is always assumed stationary,
 \be
 \Omega\neq \Omega(t)\,.
 \label{staci}
 \ee
Here, such a stationarity assumption is important.
I.a.,
this will enable us to follow, in section \ref{sectionc} below,
the methodical guidance as
provided by
a truly exceptional stationary-model study \cite{BG} by
Buslaev and Grecchi.
Indeed, Buslaev with Grecchi
were
probably the first authors who opened the question of
having the Dyson map factorized.
In this sense, their paper can be read as an immediate
predecessor and support
of our present fundamental factorizaton ansatz (cf. Eq.~(\ref{CDmap})
in section \ref{sectionc} below).
In addition, their paper can also be read as
complementing our present illustrative HF model of section \ref{sectiond}
by a
non-numerical
(i.e., rather rare)
and constructive exemplification of the preference of the extreme,
purely
OT, Dyson-map-based Hermitization (\ref{gain})
of a sufficiently realistic and, at the same time, 
mathematically friendly and analytic
ordinary-differential Hamiltonian~$H$.

%

\section{The Hilbert-space metric-operator-based
Hermitization  ${\cal H}_{auxiliary}\ \to \ {\cal H}_{physical}$\label{sectionb}}

The distinguishing feature of the
most recent implementations of the
QHQM formalism
is that the necessary
Hermitization of the observables,
i.e., the guarantee of
the unitarity of the evolution
is not provided by a reconstruction of the
textbook Hamiltonian $\mathfrak{h}$
defined as acting in a textbook Hilbert space ${\cal H}_{textbook}$
but rather by the replacement of the
Hermiticity requirement (\ref{hepo})
imposed upon $\mathfrak{h}$ in ${\cal H}_{textbook}$
by its equivalent form imposed directly upon $H$.
It is only necessary to imagine that the latter operator
is assumed introduced and defined as manifestly non-Hermitian in a manifestly unphysical
Hilbert space (say, ${\cal H}_{auxiliary}$).
In the light of its definition in terms of self-adjoint
$\mathfrak{h}$ [cf. Eq. (\ref{loss})]
it is easy to see that the conventional Hermiticity postulate (\ref{hepo})
imposed upon $\mathfrak{h}$
is formally equivalent to relation
 \be
 H^\dagger\,\Theta=\Theta\,H\,,\ \ \ \ \ \Theta=\Omega^\dagger\,\Omega\,,
 \label{quaqua}
 \ee
i.e., to the quasi-Hermiticity requirement
imposed upon $H$ in ${\cal H}_{auxiliary}$.
Now, what remains to do is to recognize that
the quasi-Hermiticity (\ref{quaqua})
of $H$ in ${\cal H}_{auxiliary}$
is equivalent to the Hermiticity of the same operator
in another, non-equivalent Hilbert space (say, ${\cal H}_{physical}$)
which only differs from ${\cal H}_{auxiliary}$ by an
appropriate {\it ad hoc\,} redefinition
of the inner product,
 \be
 ( \psi_a, \psi_b)_{auxiliary} = \br \psi_a| \psi_b\kt \ \to \
 ( \psi_a, \psi_b)_{physical} = \br \psi_a|\,\Theta\,| \psi_b\kt \,.
 \label{[6]}
 \ee
In other words, all of the mathematical
calculations may stay represented in ${\cal H}_{auxiliary}$.
For the purposes of physics (i.e., typically, for the evaluation
of predictions) the auxiliary
bra-vectors $\br \psi| \in {\cal H}_{auxiliary}'$
(where the prime $'$ marks the dual vector space \cite{Messiah})
have only to be replaced by the amended,
different bra-vector elements
 \be
 \br \psi_{physical}| = \br \psi|\Theta \equiv
 \br \psi_\Theta|\ \in \,{\cal H}_{physical}'\,
 \ee
spanning the correct and physical dual vector space ${\cal
H}_{physical}'$.
The triplet of the relevant Hilbert spaces may be
then arranged in a diagram
 \be
 \label{krexhu}
  \ba
 \begin{array}{|c|}
 \hline
  {\cal H}_{auxiliary}
 \\
  \hline
 \ea
\
 \stackrel{{\rm QHQM\ metric\ }\Theta}{ \longrightarrow }
 \
 \begin{array}{|c|}
 \hline
 {\cal H}_{physical} \\
  \hline
 \ea
\\ \ \ \
 \ \ \
\ \ \
  \stackrel{{\rm Dyson's\  map\ }\Omega}{ }
  \   \searrow
 \ \ \
\ \ \ \ \ \ \ \ \ \ \ \  \nearrow \!\!\swarrow  \
  \stackrel{{\rm equivalent\ physics\,. }\ }{ }
 \  \\ \ \  \
 \begin{array}{|c|}
  \hline
  \ \ \ \
 {\cal H}_{textbook} \ \ \ \ \\
  \hline
 \ea
 \\
\ea
 \ee
This diagram shows that the operator-transformation Hermitization (OT
Hermitization) of Eq.~(\ref{gain}) which changes the Hamiltonian finds
its alternative
in the  Hamiltonian-preserving and metric-amending Hermitization
(MA Hermitization) of Eq.~(\ref{[6]}).
In both of these scenarios,
all of
the ket- and bra-vector elements
of
the relevant Hilbert spaces
can be
represented in one, ``mathematical'' Hilbert space
${\cal H}_{auxiliary}$ in a way summarized
by the three-Hilbert-space diagrams (\ref{krexhu}) and
 \be
 \label{rexhum}
  \ba
 \begin{array}{|c|}
 \hline
 {\rm  } |\psi_m\kt,
  \ {\rm  }
 \br \psi_n | \\ 
  \hline
 \ea
\
 \stackrel{{\rm MA \ path}}{ \longrightarrow }
 \
 \begin{array}{|c|}
 \hline
 {\rm  } |\psi_m\kt,
  \ {\rm  }
 \br \psi_n |\, \Omega^\dagger\Omega \\  
  \hline
 \ea
\\  \ \ \ \ \ \
  \stackrel{{\rm OT \ path}}{ }
  \   \searrow
 \ \ \ \ \  \ \ \ \ \
  \nearrow \! \swarrow  \
  \stackrel{{\rm equivalence\,}\ }{ }\ \ \ \ \
 \  \\
 \begin{array}{|c|}
 %
  \hline
 \Omega\, |\psi_m\kt,
  \
 \br \psi_n |\, \Omega^\dagger \\  
  \hline
 \ea\ \
 \\
\ea
 \ee
(cf. also review \cite{SIGMA}).

\section{A hybrid, partial-Dyson-map plus partial-metric-operator
 Hermitization\label{sectionc}}


In our present paper we intend to complement the latter two forms of
alternative Hermitizations by a new, third, hybrid one.
A basic motivation of our proposal
resulted from the following two observations.

\begin{itemize}
\item
In the case of the Dyson-inspired
Hamiltonian-transformation Hermitization  $
H \to \mathfrak{h}=\mathfrak{h}^\dagger
 $
there exist several important weak points of the strategy.
The main one is that whenever one is able to construct
$\mathfrak{h}$ in closed form, the major part of the motivation
of working with its
non-Hermitian partner
is lost. In such a case it makes good sense to
forget about $H$ and to relocate
all of the necessary calculations directly
to ${\cal H}_{(textbook)}$.
In practice, only a few exceptions do exist:
Buslaev and Grecchi \cite{BG}, for example, found a rather exceptional
anharmonic-oscillator quantum system in which the OT strategy and
the active use of
both of the formally equivalent representation
operators $H_{({BG})}$ and $\mathfrak{h}_{({BG})}$ have been found feasible and
proved equally easy. In such a case, naturally,
the ultimate choice between the two operators
may be dictated by some complementary phenomenological criteria \cite{Mateo,Rebecca}.

\item
In the case of the strategy based on the reconstruction of the necessary
amended physical Hilbert-space metric $\Theta$
one can encounter multiple non-trivial and not always expected technical obstacles.
Their source is that whenever one
decides to treat Eq.~(\ref{quaqua}) as an implicit definition of $\Theta=\Theta(H)$,
one reveals that such a definition is not only ambiguous \cite{lotor}
but also, mostly, yielding just an approximate form of the metric:
A fairly long list of these technical MA-related obstacles
may be found presented, e.g., in review paper \cite{ali}.

\end{itemize}

 \noindent
The main idea of our present methodical innovation can be traced back to the Dyson's
physics-rooted
papers \cite{Dyson} in which the essence of the success (i.e., the
sufficiently persuasive reason for the preference of
the use of the non-Hermitian Hamiltonian model $H_{(Dyson)}$)
lied in the access to some additional information about the properties
(i.e., typically, correlations \cite{Jenssen} or clusterings \cite{Bishop}) of the
system in question.
Naturally, such a supplementary knowledge facilitated the specification of the
``good'' forms of the mapping $\Omega$.

Without the latter advantage
the probability of success
(i.e., of the simplifications achieved as a compensation of the loss of
the Hermiticity)
would be lower.
Still, we believe that the ``user-friendliness'' of the mapping $\Omega$ and/or
of the metric $\Theta(H)$ (which is, for the success, crucial)
could be also enhanced,
in the generic situation, by
a tentative factorization of the operator $\Omega$ itself.
In particular, we propose the use of the two-term ansatz
 \be
 \Omega_{} =
  \Omega_{M}\, \Omega_{H}\,.
 \label{CDmap}
 \ee
Under this assumption, the generic Dyson map
can be reinterpreted as a sequence of the
two sub-maps.
Subsequently, also the operator-transformation of the
Hamiltonian can be made sequential. In such a setting
the emerging ``interpolative'' step
 \be
 H \ \to \  H_{H} = \Omega_{H}\,H\,\Omega^{-1}_{H} \neq
 H_{H}^\dagger\,
 \label{pologain}
 \ee
might precede the ultimate OT Hermitization
 \be
 H_{H} \ \to \  \mathfrak{h}= \Omega_{M}\,H_{H}\,\Omega^{-1}_{M} =
 \Omega_{}\,H_{}\,\Omega^{-1}_{} =
 \mathfrak{h}^\dagger\,.
 \label{nepologain}
 \ee
This is, {\it in nuce}, the main message of our present paper.

One can almost immediately deduce that
the ``last step'' (\ref{nepologain}) is in fact not unavoidable.
Alternatively one may
finalize the Hermitization in the MA and
Hilbert-space-amendment spirit.
Technically, such a methodical alternative would simply mean that
one keeps the new non-Hermitian Hamiltonian $H_{H}$ unchanged.
Thus,
one recalls only the Hermiticity of $\mathfrak{h}$ and re-expresses this
property as the quasi-Hermiticity property of $H_{H}$.
In this manner,
the Hermiticity rule
$\mathfrak{h}=\mathfrak{h}^\dagger$
becomes replaced by the formally equivalent relation
 \be
 H_{H}^\dagger\,\Theta_{M}=\Theta_{M}\,H_{H}\,,\ \ \ \ \
 \Theta_{M}=\Omega_{M}^\dagger\,\Omega_{M}\,
 \label{poloqua}
 \ee
in which one just has to confirm that
the definition
 \be
 \Theta_{M}=\Omega_{M}^\dagger \Omega_{M}
 \label{[14]}
 \ee
of the new metric operator
is correct and consistent with all of the formal QHQM postulates.

Fortunately, the latter confirmation
requires just an elementary algebra
so that it may be left to the readers.
What one obtains is a hybrid-form Hermitization (HF Hermitization) of $H$.
What is, at the same time, most important is
the question about the potential usefulness of the new HF-type Hermitization
(\ref{pologain}) + (\ref{[14]})
in applications. In the next section, we will
try to persuade the readers about the 
existence of the specific merits of the HF strategy
via an
implementation of the recipe to a schematic illustrative two-level quantum system.

\section{Illustration\label{sectiond}}

For an explicit illustration of our abstract HF
representation of
the hiddenly unitary quantum evolution
let us pick up just
a maximally schematic two-by-two real-matrix model.
For the sake of simplicity let us also assume that the
Hermitian textbook Hamiltonian (with the
real and non-degenerate spectrum, say, $E_{0}=1$ and $E_{1}=2$) is,
for our methodical purposes,
diagonalized,
 \be
 \mathfrak{h}=\left[ \begin {array}{cc} 1&0\\
 \noalign{\medskip}0&2\end {array} \right]\,.
 \label{preart}
  \ee
Also, let us choose, in (\ref{CDmap}), the two Dyson-map factors (and
their product) as follows,
 \be
 \Omega_{M} =\left[ \begin {array}{cc} 1&0\\\noalign{\medskip}s&1\end {array}
 \right]\,,\ \ \ \ \
 \Omega_{H} = \left[ \begin {array}{cc} 1&t\\\noalign{\medskip}0&1\end {array}
 \right]\,,\ \ \ \ \
 \Omega_{} =\left[ \begin {array}{cc} 1&t\\\noalign{\medskip}s&st+1\end {array} \right]
 \,.
 \label{twocd}
 \ee

In the light of Eq.~(\ref{loss}) we may replace our
$\mathfrak{h}$ of Eq.~(\ref{preart}) by its manifestly non-Hermitian
isospectral alternative
 \be
 H=
 \left[ \begin {array}{cc} -st+1&- \left( st+1 \right) t
 \\
 \noalign{\medskip}s&st+2\end {array} \right]\,.
 \label{ohjoj}
 \ee
Obviously, in a sharp
contrast to trivial Eq.~(\ref{preart}),
the new Hamiltonian-sampling
matrix $H$ looks ``complicated''.
For our present methodical purposes it is important that
the latter observation can be also reinterpreted as a statement that
the contrast of complexities will survive
in any other basis, with an optimal representation relocated,
including the opposite extreme case in which the matrix form of 
$H$ would
look decisively simpler
than the matrix form of $\mathfrak{h}$.

For the sake of simplicity let us just keep in mind the latter
methodical argument while using the illustrative matrices,
due to their maximal pedagogical appeal, in their fixed and given form.
This means that we may still speak about a generic nature of
the sharp contrast between the
triviality of the secular equation for one of the models (here, for $\mathfrak{h}$)
and an apparently more complicated form of its
Dyson-mapping-mediated analogue.

In our particular example, therefore, one is able to estimate the complexity
of the matrix model
only after some explicit
algebraic manipulations.
These manipulations also enable us to
reveal 
and demonstrate that
the apparently complicated and $s-$ and
$t-$dependent secular
polynomial $\det(H-E)$
associated with $H$ is in fact
also elementary and  $s-$ and
$t-$independent and easily factorizable,
$\det(H-E)= {E}^{2}-3\,E+2=(E-1)\,(E-2)$.

Another apparent
contrast in complexity between $\mathfrak{h}$ and $H$
emerges during the verification
of the validity of the MA-related quasi-Hermiticity relation (\ref{quaqua}).
Indeed, after insertion we obtain the most complicated matrix result
 \be
 H^\dagger \Theta=\Theta\,H=\left[ \begin {array}{cc} 1+2\,{s}^{2}&t+2\,t{s}^{2}+2\,s
 \\\noalign{\medskip}t+2\,t{s}^{2}+2\,s&{t}^{2}+2\,{t}^{2}{s}^{2}+4\,ts
 +2
 \end {array} \right]\,.
 \ee
Indeed, it takes time to see that
this matrix is Hermitian and positive definite at any positive
values of the real parameters $s$ a $t$: Incidentally, in our recent paper \cite{EPJP}
the operator products $H^\dagger \Theta=\Theta\,H$
were denoted by a dedicated symbol $Y$ and studied {\it per se}.
Incidentally, such a study was well motivated since these products
prove highly relevant
in the context of the random matrix theories \cite{Joshua}.
Naturally, in
such an applied physics context their 
complexity or simplicity would be an important 
characteristics, indeed.

In any physical and QHQM context, after all,
one finally has to turn attention to the explicit form of the
MA-mediating
Hermitizing metric. In our model we have
 \be
  \Theta_{}=
 \left[ \begin {array}{cc} 1+{s}^{2}& \left( 1+{s}^{2} \right) t+s\\
 \noalign{\medskip}t+ \left( ts+1 \right) s& \left( t+ \left( ts+1
 \right) s \right) t+ts+1\end {array} \right]
 \label{[19]}
 \ee
Indeed, one might think that the necessary (and,, quite often, difficult
\cite{trialy}) proofs of its positive definiteness
can still be done, quite routinely, via the computer-assisted
symbolic manipulations.
In this sense we encountered a technical surprise
because in spite of a virtual triviality of the underlying two-by-two-matrix algebra
the explicit evaluation of the spectrum $\{\theta_{\pm}\}$ of the metric
as provided by the routines of
MAPLE \cite{MAPLE} appeared complicated,
 \be
 \theta_{\pm}=1/2\,{t}^{2}+1/2\,{t}^{2}{s}^{2}+ts+1+1/2\,{s}^{2}
 \pm 1/2\, \sqrt {D}
 \ee
containing, in addition, an unexpectedly long polynomial
 \be
 D={2\,{t}^
 {4}{s}^{2}+4\,{t}^{3}s+4\,{t}^{3}{s}^{3}+{t}^{4}{s}^{4}+{s}^{4}+{t}^{4
 }+4\,{s}^{2}+8\,ts+4\,{t}^{2}+10\,{t}^{2}{s}^{2}+2\,{t}^{2}{s}^{4}+4\,
 t{s}^{3}}\,.
 \label{expr}
 \ee
Still, one arrives at a nice result.

\begin{lemma}
\label{lemma1}
The matrix of metric (\ref{[19]}) is positive definite.
\end{lemma}
\begin{proof}
The proof of the positivity of the metric
remained, fortunately, easy and straightforward because
after the check of the
calculation which was still feasible to
perform by hand the helpful key simplification
\be
D= { \left( 4+{t}^{2}{s}^{2}+({s}+{t})^{2} \right)  \left( {t}^{2}
 {s}^{2}+({s}+{t})^{2} \right)}
\ee
of the long expression (\ref{expr})
by factorization was immediately visible.
\end{proof}


We believe that the latter set of formulae persuaded the readers that
the technical obstacles encountered during the pure-form OT or MA
Hermitizations need not be inessential. We believe that the hybrid approach
offering, via Eq.~(\ref{pologain}), the next-to elementary Hamiltonian
 \be
 H_{H}=\left[ \begin {array}{cc} 1&0\\\noalign{\medskip}s&2\end {array} \right]\,
 \ee
in combination
with the next-to elementary inner-product metric of Eqs.~(\ref{poloqua})
and~(\ref{[14]}), viz.,
 \be
  \Theta_{M}=
 \left[ \begin {array}{cc} 1+{s}^{2}&s\\\noalign{\medskip}s&1\end {array} \right]
  \label{Z21}
 \ee
delivers a fairly persuasive argument in favor of the optimality
of our newly proposed third, HF Hermitization strategy.
{\it Pars pro toto\,} it is worth noting that in contrast to
the MA-related Lemma~\ref{lemma1}
the analogous HF-validating proof of the
positive definiteness of metric (\ref{Z21}) at all $s$
proceeds via the positivity of its much more elementary eigenvalues
$\ \sim 2+s^2 \pm \sqrt{(2+s^2)^2-4}$.

In our final argument supporting the innovative HF Hermitization strategy
we may say that the single-parameter dependence of the
HF-representing pair of operators $H_{H}$ and $\Theta_{M}$
symbolizes and underlines a half-way formal character of the HF approach.
We see that also by the criteria of complexity of formulae it lies
in between the two (viz., OT and MA) extremes, one of which is always,
in our schematic illustration, two-parametric.

\section{Conclusions}

We may summarize our present proposal of an innovative HF
reformulation of quantum mechanics as a composition 
of the partial OT mapping
$H \to H_H$
of Eq.~(\ref{pologain}) with the partial MA
quasi-Hermiticity constraint (\ref{poloqua}) which is
formulated in terms of the
reduced HF metric $\Theta_M$.
In conclusion let us now add a few complementary comments
on such a form of the model-building strategy.

First of all, let us underline the importance of the
stationarity assumption (\ref{staci}).
It represents a key to the
practical applicability and feasibility of the 
present HF Hermitization.
In the language of mathematics the stationarity is
really the constraint which
enabled us to make the full use of our fundamental
Dyson-map-factorization ansatz (\ref{CDmap}).
Now, we just have to add that in the nearest future we plan to
generalize and develop the HF theory also
towards its non-stationary extensions.
In fact, the first steps in this direction were already made,
in 2008, in
our
three-Hilbert-space description of the unitary quantum evolution
as proposed in paper \cite{timedep} and as reviewed, later, in
\cite{SIGMA}.
On these grounds we were
able to formulate
a non-Hermitian (or, better, hiddenly Hermitian)
generalization of the
standard unitary quantum mechanics in
interaction picture \cite{NIP}  (see also \cite{Frith})
and, in the special cases, not only
in 
Heisenberg picture \cite{Heisenberg} (see also 
the parallel developments in \cite{Heisenbergb})
but also
in the generalized Heisenberg picture (see \cite{ITJP}
and also \cite{Frith,Frithb}).

In all of the latter papers
the Dyson maps were assumed
manifestly non-stationary, $\Omega=\Omega(t)$.
Interested readers should certainly consult the most recent
forms of the non-stationary QHQM theory, say,
in the comprehensive reviews \cite{revju,compa}.
Naturally, due to the non-stationarity of these innovations
there exists no direct comparison
with our present, stationary version of the HF approach.
Nevertheless, it is
worth adding that the authors of review \cite{compa}
proposed to call the non-stationary Dyson maps ``generalized vielbeins''.
This is a truly remarkable detail because it
combines the apparently purely terminological decision with a truly
deep insight in the possible applicability
of the non-stationary Dyson maps in quantum cosmology.
Unfortunately, such an observation and
remark already go fairly beyond the scope of our present letter.
Still,
interested readers are recommended to have a look
at the corresponding newest considerations in \cite{WDW}
or in the very recent preprint \cite{compb}.

In our second, independent and last 
concluding comment on our present HF method
let us point out that 
it is, in some sense, surprising that in the newest literature on the
non-Hermitian quantum Hamiltonians (cf., e.g., \cite{Christodoulides,Carlbook})
the ``mainstream'' attention is slowly moving 
from the closed-system theory
to the 
studies of the open and nonlinear quantum systems.
In this case, naturally, the unitarity
ceases to be an issue. One of the possible explanations of the shift
could be
sought in the fact that in spite of the initially enthusiastic
acceptance of the Hermitization ideas, their subsequent
more detailed analysis and implementations encountered multiple
(and also not quite expected)
conceptual as well as technical complications \cite{MZbook,Jones,Siegl}.
Not enough attention has been paid to the quantum systems
(sampled, e.g., by the Buslaev's and Grecchi's paper \cite{BG})
in which a way out of the difficulties could be sought 
in a decomposition of the Dyson map into a sequence of
its simpler factors. In this sense, our present paper on Hermitizations
based on the composite Dyson map ansatz (\ref{CDmap})
can be read as an encouragement of a renewal of interest in unitarity and of a 
start of a new active research in the apparently old-fashioned closed-system direction.

%
%
%
%
%
%
%
%
%
%
%
%

\end{document}